\documentclass[11pt, a4paper]{article}
\usepackage{enumitem}
\usepackage{hyperref}
\usepackage{authblk}
\usepackage{tikz}
\usepackage{tikz-cd}
\usepackage{amsmath}
\usepackage{amsfonts}
\usepackage{amsthm}
\usepackage{stmaryrd}
\usepackage{graphicx}

\usepackage{tikz}
\usetikzlibrary{arrows,automata,positioning}

\newtheorem{example}{Example}

\def\calA{{\mathcal A}}

\def\N{{\sym N}}

\def\Z{{\sym Z}}

\def\kk{\mathbf{k}}
\def\Set{\mbox{\textbf{Set}}}

\def\up#1{\raise 1ex\hbox{\footnotesize#1}}
\def\pointir{\unskip . --- \ignorespaces}

\def\Set{\mbox{\textbf{Set}}}
\def\FinSet{\mbox{\textbf{FinSet}}}

\def\Mon{\mbox{\textbf{Mon}}}

\def\Grp{\mbox{\textbf{Grp}}}

\def\AAU#1{#1\mbox{\textbf{-AAU}}}

\def\Lie#1{#1\mbox{\textbf{-Lie}}}

\def\N{\mathbb{N}}
\def\Z{\mathbb{Z}}
\def\f#1{\mathfrak{#1}}
\def\botimes{\bigotimes}
\newcommand{\calB}{{\mathcal B}}
\newcommand{\calC}{{\mathcal C}}

\newcommand{\calU}{{\mathcal U}}

\newtheorem{theorem}{Theorem}
\newtheorem{lemma}{Lemma}

\newtheorem{corollary}{Corollary}
\newtheorem{remark}{Remarks}
\newtheorem{notation}{Notations}
\newcommand{\bigslant}[2]{{\raisebox{.2em}{$#1$}\left/\raisebox{-.2em}{$#2$}\right.}}

\title{About enveloping algebras of direct sums.}
\author[1]{G\'erard H. E. Duchamp\thanks{gheduchamp@gmail.com}}
\author[2]{Jean-Gabriel Luque\thanks{jean-gabriel.luque@univ-rouen.fr}} 
\author[1]{Christophe Tollu\thanks{ct@lipn.univ-paris13.fr}}
\author[3]{Vu Nguyen Dinh\thanks{ndvu@math.ac.vn}}
\author[4]{Various Authors\thanks{toadd}}
\affil[1]{\small{LIPN, Northen Paris University, Sorbonne Paris City, 93430 Villetaneuse, France.}}
\affil[2]{\small{GR$^2$IF,  Université de Rouen, Avenue de l’Université, 76801 Saint-Étienne du Rouvray,\\
Cedex, France.}}
\affil[3]{\small{USTH, University of Science and Technology of Hanoi, Vietnam.}}
\date{\today}

\begin{document}
\maketitle

\begin{abstract}
We solve the PBW-like problem of normal ordering for enveloping algebras of direct sums.  
\end{abstract}
\section{A question about the enveloping algebra of a direct sum.}
This question is imported from \cite{MO-Orig}. It is linked to this one 
\cite{MO-Smash} in the case of semi-direct products. 

Let us consider a Lie $\kk$-algebra ($\kk$ is a commutative ring) written as a (module) direct sum of two of its subalgebras
$$
\mathfrak{g}=\mathfrak{g}_1\oplus\mathfrak{g}_2\ (\oplus=\oplus_{\kk-mod})
$$ 
and the linear maps associated to this direct sum decomposition
\begin{equation}\label{j-and-p}
\begin{tikzcd}[column sep=1.5cm]
\f{g}_i \arrow[r, shift left, "j_i"]
\arrow[r, <-, shift right, swap, "p_i"] & \f{g}
\end{tikzcd}
\mbox{ such that }j_1p_1 + j_2p_2=Id_{\f{g}}
\end{equation}
($Id_{\f{g}}$ is a sum of two orthogonal projectors, remark 
that only $j_i$'s are Lie morphisms in general).\\
We get, at once, the maps $\calU(j_i)$ through the universal algebra functor $\calU$ (see below section \ref{Universal} ‘‘Universal Constructions'') as follows. 
\begin{equation}\label{spp1}
\begin{tikzcd}[column sep=2cm, row sep=1.2cm]
\f{g}_i
\arrow[d,"\sigma_i",swap] \arrow[r,"j_i"]
& \f{g} 
\arrow[d,"\sigma",swap]\\  
\calU(\f{g}_i)\arrow[r,"\calU(j_i)"] 
& \calU(\f{g})
\end{tikzcd}
\end{equation}
allowing us to multiply members of 
$\mathcal{U}(\mathfrak{g}_1)\otimes_{\kk}\mathcal{U}(\mathfrak{g}_2)$ within 
$\mathcal{U}(\mathfrak{g})$ by the composite map 
\begin{equation}\label{map1}
\mu_{state}=\mu\circ (\calU(j_1)\otimes_{\kk} \calU(j_2)):\ \mathcal{U}(\mathfrak{g}_1)\otimes_{\kk}\mathcal{U}(\mathfrak{g}_2)\to\mathcal{U}(\mathfrak{g})
\end{equation}
where $\mu$ is the multiplication of $\mathcal{U}(\mathfrak{g})$.

One can check, using generators, that $\mu_{state}$ is surjective (and, in many usual cases \cite{MO-Orig} bijective). 

\textbf{Question:} Is the property that $\mu_{state}$ is bijective, true in the general case ?  
\begin{remark}\label{Quot-Rem}
\begin{enumerate}
\item Unless explicitly stated, all tensor products will be understood over $\kk$. 
\item All states below are elements of spaces of the form 
$M=\calA_1\otimes \calA_2$ (see Eq. \ref{Domains-arrows1}) where, for $i=1,2$, $(\calA_i,\ast_i)$ is a $\kk$-AAU (i.e. Associative Algebra with Unit). Such a space $M$ is naturally a $\calA_1-\calA_2$ bimodule ($\calA_1$ module by multiplication on the left and $\calA_2$ module by multiplication on the right). To avoid confusion (as tensors may appear inside $\calA_1$), the separating tensor of $M=\calA_1\botimes \calA_2$ will, from time to time, be noted bold and oversized. 
\label{States-bimod}
\end{enumerate}
\end{remark}
\begin{example}\label{Example-Generator}
\begin{enumerate}
\item In the remarks and examples below, for any $\Z$-algebra $\calA$ and 
$q\in \N_{\ge2}$, we will 
note $(\calA)_q$ the quotient $\Z$-algebra $\bigslant{\calA}{q.\calA}$ (in all cases $q.\calA$ is an ideal).\\  
For example, the situation $\f{g}=\f{g}_1\oplus\f{g}_2$ where no factor is an ideal is frequent for Lie algebras admitting a triangular decomposition\footnote{As Kac-Moody algebras \cite{MP}, have also a look at \cite{VK} Ch 1 Exercise 1.8 (local Lie algebras).}%
$\f{g}=\f{n}_+\oplus\f{h}\oplus\f{n}_-$ for example 
$$
\f{sl}(n,\Z)=T_+(n,\Z)\oplus D(n,\Z)\oplus T_-(n,\Z)
$$ 
and one can create an example without any basis with $\kk=\Z,\ q\in \N_{\ge2}$ with 
\begin{eqnarray}
&&\f{g}=\f{sl}(n,\Z)=T_+(n,\Z)\oplus D(n,\Z) \oplus T_-(n,\Z) \mbox{ then }\cr
&&(\f{g})_q=T_+(n,\Z)_q\oplus D(n,\Z)_q\oplus \left(T_-(n,\Z)\right)_q 
\mbox{ and, if one needs two factors, }\cr
&&(\f{g})_q=\left(T_+(n,\Z)\oplus D(n,\Z)\right)_q\oplus \left(T_-(n,\Z)\right)_q
\end{eqnarray}
\end{enumerate}
\end{example}
{\renewcommand{\abstractname}{Acknowledgements}
\begin{abstract}
We thank Darij Grinberg, from Drexel University and Jim Humphreys (through MathOverflow) for fruiful interactions and their interest for this question. 
\end{abstract}
}
\section{Universal Constructions}\label{Universal}
\subsection{General principle.}\label{FreeObj}
In this subsection, we introduce the combinatorial (free) objects that we will use throughout the manuscript and the notion of enveloping algebra of a Lie algebra. These objects (call them $G(X)$) together with a map $j_X:\ X\to G(X)$ are all solutions of universal problems. We will recall the definition, notation and terminology about these free\footnote{Or freely generated in the case of enveloping algebras.} objects below (cf. in general Bourbaki \cite{B-ToS} 
Ch IV \S 3 or \cite{McLane} and, in particular,  
\cite{B-Alg13} Ch I \S 7.1 and Lothaire \cite{LothCOW} Prop 1.1.1 for words and the free monoids, Bourbaki \cite{B-Lie13} Ch II \S 2.2 Prop 1 and Reutenauer \cite{FLA} Thm 0.4 for free Lie algebras and Bourbaki \cite{B-Alg13} for enveloping algebras i.e. towards the free associative algebras with unit and Dinh Vu Nguyen's thesis \cite{NDVU-PhD} for all these matters), but here, we state the general principle.\\
The scheme is the same for all categories considered in the following list ($\kk$ being a fixed ring).
\begin{equation}\label{CatList}
\Mon,\Grp,\AAU{\kk},\Lie{\kk}
\end{equation} 

\smallskip
All objects of these categories can be considered as sets, we then have a natural ‘‘forgetful'' functor $F$ such that, $\calA$ being an object (of one of these categories), 
$F(\calA)$ is the set underlying the structure $\calA$. As well, any 
$\calA\in \AAU{\kk}$ can be considered as a Lie algebra with the bracket 
$[x,y]:=xy-yx$

We are now in the position 
of stating the universal problem leading to left-adjoint of a functor $F$.

\smallskip\noindent
\textbf{Universal problem (w.r.t. $F$, naive version\footnote{See the theory of Heteromorphisms\cite{Hets}})}\pointir\\
\textit{For any set $X$ ($\calC$ being one of the categories as above) does there exist a 
pair $(j_X,G(X))$ ($G(X)$ being an object of $\calC$ and $j_X:\ X\to G(X)$ an heteromorphism) such that:\\
For any map $f:\ X\to \calA$ (heteromorphism), there exists a unique 
$\widehat{f}\in Hom_{\calC}(G(X),\calA)$ such that $f=F(\widehat{f})\circ j_X$}.

\smallskip
\begin{remark}\label{FreeFunct} i) It might happen that $G$ be not defined everywhere as shows the case with $\calC=\FinSet$, $F$ being the inclusion functor (i.e. $F(X)=X$ for every finite set and $F(f)=f$ for every set-theoretical map between finite sets).\\ 
However a solution of the universal problem \eqref{adjunction1}, \textit{for all} $X$, provides a \textit{ free functor } $G: \Set \to \calC, X \mapsto G(X) $ which is left-adjoint to the forgetful 
functor $F: \calC \to \Set$. The reader must be aware that, in general, the notion of ‘‘forgetful functor'' (here constructed from algebraic structures and sets) 
is informal.
\begin{equation}\label{adjunction1}
\begin{tikzcd}[column sep=1.6cm, row sep=0.8cm]
\Set\arrow[rr,leftarrow,"F",pos=0.5,densely dotted]
    \arrow[dd,dash,
    start anchor={[xshift=17ex, yshift=4ex]},
    end anchor={[xshift=17ex]}
    ]&
& \calC
\\[-20pt] 
X \arrow[rr,"f",pos=0.3,dashed]
\arrow[rrd,dashed, "j_X",pos=0.41,swap]&      &  
\calA\\
\phantom{x}&& G(X). \arrow[u, "\widehat{f}"']
\end{tikzcd}
\end{equation}
ii) We recall here that the universal enveloping algebra of a Lie $\kk$-algebra 
$\mathfrak{g}$ is a pair $(\sigma,\calU(\mathfrak{g}))$, where  $\calU(\mathfrak{g})$ is an object in $\AAU{\kk}$ and 
$\sigma:\ \mathfrak{g}\to \calU(\mathfrak{g})$ is a morphism in $\Lie{\kk}$,
which is a solution of the following universal problem:
\begin{equation}\label{EnvUnivDiag}
\begin{tikzcd}[column sep=1.6cm, row sep=0.8cm]
\Lie{\kk}\arrow[rr,leftarrow,"F",pos=0.5,densely dotted]
    \arrow[dd,dash,
    start anchor={[xshift=17ex, yshift=4ex]},
    end anchor={[xshift=17ex]}
    ]&
& \AAU{\kk}\\[-20pt] 
\mathfrak{g} \arrow[rr,"f",pos=0.3,dashed]
\arrow[rrd,dashed, "\sigma",pos=0.41,swap]&      &  
\calA\\
\phantom{x}&& \calU(\mathfrak{g}). \arrow[u, "\widehat{f}"']
\end{tikzcd}
\end{equation}
From this arises that there exists the universal enveloping functor 
\begin{eqnarray}
\calU: \Lie{\kk} \rightarrow \AAU{\kk}, && \mathfrak{g} \longmapsto  
(\sigma,\calU(\mathfrak{g}))
\end{eqnarray}
which is a left-adjoint to the Lie-ation functor $F$. 
\end{remark}
\begin{notation}\label{spare-pieces-1}
In the following, we will use notations as above and also 
\begin{enumerate}
\item Identity of $\f{g_i}$ (resp. of $\calU(\f{g}_i)$) will be noted, for short, $I_i$ (resp. $I_{\calU_i}$)
\item  The maps $\sigma_i\ \mathfrak{g}_i \mapsto \calU(\mathfrak{g}_i)$ 
(resp. the map $\sigma:\ \mathfrak{g} \mapsto \calU(\mathfrak{g})$), 
%
\item The maps deduced by universal constructions, as in the preamble, the maps 
$\calU(j_i)$ 
\begin{equation}\label{spp1}
\begin{tikzcd}[column sep=2cm, row sep=1.2cm]
\f{g}_i
\arrow[d,"\sigma_i",swap] \arrow[r,"j_i"]
& \f{g} 
\arrow[d,"\sigma",swap]\\  
\calU(\f{g}_i)\arrow[r,"\calU(j_i)"] 
& \calU(\f{g})
\end{tikzcd}
\end{equation}
and $\psi:\ T(\f{g})\to \calU(\f{g})$ (resp. 
$\psi_i:\ T(\f{g}_i)\to \calU(\f{g}_i)$) the natural (quotient) maps (see Bourbaki \cite{B-Lie13} Ch I \S 2.7 p17).
\item The chaining of domains and maps involved is as follows
\begin{equation}\label{Domains-arrows1}
\begin{tikzcd}[column sep=1.9cm]
T(\f{g}_1)\otimes \calU(\f{g}_2)
\arrow[r,"\psi_1\otimes I_{U_2}"]
    \ar[to path={ -- ([yshift=-6ex]\tikztostart.south) -| (\tikztotarget)},
    rounded corners=12pt]{rrr}
& \calU(\f{g}_1)\otimes \calU(\f{g}_2)
\arrow[r,"\calU(j_1)\otimes\, \calU(j_2)"]
    \ar[to path={ -- ([yshift=-4ex]\tikztostart.south) -| (\tikztotarget)},
    rounded corners=12pt]{rr}
& \calU(\f{g})\otimes \calU(\f{g})
\arrow[r,"\mu"]
& \calU(\f{g})
\end{tikzcd}
\end{equation}
The upper long arrow being 
$\mu_{state}:=\mu_{state}^U=\mu\circ (\calU(j_1)\otimes\, \calU(j_2))$, 
the lower long (longer) arrow being 
$\mu_{state}^T=\mu_{state}\circ (\psi_1\otimes I_{U_2})$. We will now construct the normal form calculator and, from it, deduce a section $s$ of $\mu_{state}^U$ which will turn out to be bijective. Knowing already that $\mu_{state}^U$ is surjective, it will be sufficient to establish that 
$$
s\circ \mu_{state}^U=Id_{\calU(\f{g}_1)\otimes\, \calU(\f{g}_2)}
$$    
\end{enumerate}
\end{notation}
\section{Step-by-Step construction of a normal form calculator}\label{Normal-Form-Calc}
Having remarked that the domain and codomain of $\mu_{state}\ (=\mu_{state}^U)$ are 
$\calU(\f{g}_1)-\calU(\f{g}_2)$ modules ($\calU(\f{g}_1)-$ by multiplication on the left and $-\calU(\f{g}_2)$ by multiplication on the right for the domain and through $\calU(\f{g}_1))\otimes \calU(\f{g}_2))$ for the codomain) our strategy will be to construct a $-\calU(\f{g}_2)$ section of $\mu$ (this linearity will help us 
to make the construction, at first defined on $T(\f{g}_1)\otimes \calU(\f{g}_2)$, pass to quotients). We observe now, in all cases when $\mu_{state}$ is one-to-one, there is an action on the left ($g\ast_{\calU}$) of $\calU(\f{g})$ on the space of states 
$\calU(\f{g}_1)\otimes \calU(\f{g}_2)$ provided by transport of structure \cite{TransStruct} as follows  
\begin{equation}
g\ast_{\calU} (m_1\otimes m_2)=\mu_{state}^{-1}(g.\mu_{state} (m_1\otimes m_2))
\end{equation}
Now, we will construct this action in the general case by passing to quotients a similar action on $T(\f{g}_1)\otimes \calU(\f{g}_2)$ denoted by $g\ast_T-$. This compatibility (to be proved) is illustrated by the following diagram
\begin{equation}\label{sqr1}
\begin{tikzcd}[column sep=2cm, row sep=1.2cm]
T(\f{g}_1)\otimes \calU(\f{g}_2)
\arrow[d,dashed," g\ast_T- ",swap] \arrow[r,"\psi_1\otimes I_2"]
& \calU(\f{g}_1)\otimes \calU(\f{g}_2) 
\arrow[d,dashed," g\ast_{\calU}- ",swap]\\  
T(\f{g}_1)\otimes \calU(\f{g}_2)\arrow[r,"\psi_1 \otimes I_2"] 
& \calU(\f{g}_1)\otimes \calU(\f{g}_2) 
\end{tikzcd}
\end{equation}
where $\psi_1:\ T(\f{g}_1)\to \calU(\f{g}_1)$ is the natural (quotient) map and $I_2=Id_{\calU(\f{g}_2)}$.
We will proceed in fours steps 
\begin{enumerate}
\item Construction
\item Compatibility with $\equiv_{\psi_1}$
\item Action $\ast_{\calU}$ as a Lie action.
\item Section and isomorphism
\end{enumerate}
\subsection{Construction of the actions $g\ast$}\label{actions1}
Let us recall that $\mathfrak{g}$ is a Lie algebra split ($\kk$-module decomposition) as follows 
\begin{equation}\label{G-split}
\mathfrak{g}=\mathfrak{g}_1\oplus\mathfrak{g}_2\ (\mbox{here }
\oplus=\oplus_{\kk-mod})
\end{equation}
Let $j_i,p_i$ be the corresponding embeddings and projectors (see also the end of paragraph \ref{FreeObj}). In addition, we will note $\psi_1$ the morphism of $\kk$-AAU $\psi_1:\ T(\f{g}_1)\to \calU(\f{g}_1)$ obtained by multiplication of factors and $\calU(j_2)$, the natural morphism 
$\calU(j_2):\ \calU(\f{g}_2)\to \calU(\f{g})$ as defined above (see diagram \ref{spp1})
We now have the following 
\begin{theorem}\label{Th-A}
With the notations as above, \\
i) there exists a unique linear map 
\begin{equation}
\Phi:\ \f{g}\otimes T(\f{g}_1)\otimes \calU(\f{g}_2)\to 
T(\f{g}_1)\otimes \calU(\f{g}_2).
\end{equation} 
(in the sequel, $\Phi(g\otimes t\botimes m)$ will be alternatively noted $g\ast_T (t\botimes m)$)\\ 
such that
\begin{eqnarray}\label{rec-PhiT}
\left\{\begin{array}{l}
g\ast_T (1_{T(\f{g}_1)}\otimes m)=p_1(g)\otimes m + 
1_{T(\f{g}_1)}\otimes \sigma_2 p_2(g).m\mbox{ for all } (g,m)\in \f{g}\times \calU(\f{g}_2)
\\[3mm]
g\ast_T(x\otimes t\botimes m)=[g,j_1(x)]\ast_T(t\botimes m) + x\otimes 
\big(g\ast_T(t\botimes m)\big)\\[2mm]
\mbox{ for all } (g,x,t,m)\in \f{g}\times\f{g}_1\times T(\f{g}_1)\times \calU(\f{g}_2)    
\end{array}\right.
\end{eqnarray}
(Nota : For the sake of clarity, we have used the blue tensor product as explained in\\ Remark \ref{Quot-Rem}.\ref{States-bimod}.)\\
ii) This map is filtered in the following sense
\begin{equation}
\Phi\big(\f{g}\otimes T_{\leq n}(\f{g}_1)\otimes \calU(\f{g}_2)\big)\subset 
T_{\leq n+1}(\f{g}_1)\otimes \calU(\f{g}_2)
\end{equation}
iii) It is compatible with\\
a) The $\calU(\f{g}_2)$ right module structure of 
$T(\f{g}_1)\otimes \calU(\f{g}_2)$ as 
\begin{eqnarray}
&&\mbox{ for }(g,t,m)\in \f{g}\times T(\f{g}_1)\times \calU(\f{g}_2)\mbox{ one has }\cr
&& g\ast_T (t\otimes m)=\big(g\ast_T (t\otimes 1_{\calU(\f{g}_2)})\big).m
\end{eqnarray}
b) Multiplication of factors. Let  
\begin{equation}
\mu_{state}^T=\mu\circ (\calU(j_1)\otimes \calU(j_2))\circ(\psi_1\otimes I_2)
\end{equation}
(see Notation \eqref{spare-pieces-1} and Eq. \eqref{map1}) as, for all 
$(g,t,m)\in \f{g}\times T(\f{g}_1)\times \calU(\f{g}_2)$, we have 
\begin{equation}\label{mu-compat}
\mu_{state}^T(g\ast_T(t\otimes m))=\sigma(g).\mu_{state}^T(t\otimes m)=
\sigma(g).\mu_{state}^T(t\otimes 1_{\calU(\f{g}_2}).\,\calU(j_2)(m)
\end{equation}
iv) There is a unique map 
\begin{equation}
\Phi_U:\ \f{g}\otimes \calU(\f{g}_1)\otimes \calU(\f{g}_2)\to 
\calU(\f{g}_1)\otimes \calU(\f{g}_2)
\end{equation} 
such that the following diagram commutes 
\begin{equation}\label{Cell1}
\begin{tikzcd}[column sep=2cm, row sep=0.8cm]
\f{g}\otimes T(\f{g}_1)\otimes\calU(\f{g}_2)\arrow[d,"\Phi",swap]
\arrow[r,"I_1\otimes\psi_1\otimes I_{U_2} "] 
& \f{g}\otimes \calU(\f{g}_1)\otimes \calU(\f{g}_2)\arrow[d," \Phi_U",swap]\\  
T(\f{g}_1)\otimes \calU(\f{g}_2) \arrow[r,"\psi_1\otimes I_{U_2}"] 
& \calU(\f{g}_1)\otimes \calU(\f{g}_2)
\end{tikzcd}
\end{equation} 
\end{theorem}
\begin{proof}
\begin{enumerate}
\item[i)] and ii) We will show, by induction on $n$, the following statement:\\
For all $n\ge0$, there exists a unique linear map 
\begin{equation}
\Phi_n:\ \f{g}\otimes T_{\leq n}(\f{g}_1)\otimes \calU(\f{g}_2)\to 
T_{\leq n+1}(\f{g}_1)\otimes \calU(\f{g}_2)
\end{equation} 
noted $g\ast_T^{(n)} (t\otimes m):=\Phi_n(g\otimes t\otimes m)$
such that, 
\begin{eqnarray}\label{recn}
\left\{\begin{array}{l}
g\ast_T^{(n)}(1_{T(\f{g}_1)}\otimes m)=p_1(g)\otimes m + 
1_{T(\f{g}_1)}\otimes \sigma_2 p_2(g).m\mbox{ for all } (g,m)\in \f{g}\times \calU(\f{g}_2)
\\[3mm]
g\ast_T^{(n)}(x\otimes t\otimes m)=[g,j_1(x)]\ast_T^{(n-1)}(t\otimes m) + x\otimes 
\big(g\ast_T^{(n-1)}(t\otimes m)\big)\\[2mm]
\mbox{ for all } (g,x,t,m)\in \f{g}\times\f{g}_1\times T_{\leq n-1}(\f{g}_1)\times \calU(\f{g}_2)    
\end{array}\right.
\end{eqnarray}
For $n=0$, $\Phi_{0}$ is clearly uniquely defined by 
\begin{equation*}
\Phi_{0}(g\otimes \lambda.1_{T(\f{g}_1)}\otimes m):=
\lambda.\big(p_1(g)\otimes m + 
1_{T(\f{g}_1)}\otimes \sigma_2 p_2(g)m\big).
\end{equation*}
We now suppose $(\Phi_{j})_{0\leq j\leq n}$ to be uniquely defined by  
\eqref{recn} and show the same for some\footnote{which will turn out to be unique.}
\begin{eqnarray}
&&\Phi_{n+1}:\ \f{g}\otimes T_{\leq n+1}(\f{g}_1)\otimes \calU(\f{g}_2)\to 
T_{\leq n+2}(\f{g}_1)\otimes \calU(\f{g}_2) \mbox{ with }\cr
&& g\ast_T^{(n+1)} (t\otimes m):=\Phi_{n+1}(g\otimes t\otimes m)
\end{eqnarray}

Remarking that 
$$
\f{g}\otimes T_{\leq n+1}(\f{g}_1)\otimes \calU(\f{g}_2)=
\f{g}\otimes T_{\leq n}(\f{g}_1)\otimes \calU(\f{g}_2) \oplus
\f{g}\otimes T_{n+1}(\f{g}_1)\otimes \calU(\f{g}_2)
$$
we define $\Phi_{n+1}$ as coinciding with $\Phi_{n}$ on the sector 
$\f{g}\otimes T_{\leq n}(\f{g}_1)\otimes \calU(\f{g}_2)$.\\ 
Now for 
$$
(g,x,t,m)\in \f{g}\times\f{g}_1\times T_{n}(\f{g}_1)\times \calU(\f{g}_2), 
$$
we observe that
$$
(g,x,t,m)\mapsto [g,x]\ast_T (t\otimes m) + x\otimes (g\ast_T (t\otimes m))
$$ 
is $\kk$-quadrilinear which entails existence and unicity of a linear map 
$$
\check{\Phi}_{n+1}:\ \f{g}\otimes \big(\f{g}_1\otimes T_{n}(\f{g}_1)\big)\otimes \calU(\f{g}_2) = 
\f{g}\otimes T_{n+1}(\f{g}_1)\otimes \calU(\f{g}_2) \to 
T_{\leq n+2}(\f{g}_1)\otimes \calU(\f{g}_2) 
$$
This allows us to set $\Phi_{n+1}=\Phi_{n}\oplus \check{\Phi}_{n+1}$ which is uniquely defined due to \eqref{recn}\footnote{In fact, $\Phi$ is the inductive limit of the sequence $\Phi_n$.}.  
\item[iii)] a) Again, by induction.\\
b) Compatibility with $\mu_{state}^T$.\\ 
Again, we prove this by induction on $n$ on the property that, for all 
$(g,t,m)\in \f{g}\times T_n(\f{g}_1)\times \calU(\f{g}_2)$, we have \eqref{mu-compat}.\\ 
For $n=0$, it suffices to remark that 
\begin{eqnarray}
&&\mu_{state}^T\big(g\ast_T(1_{T(\f{g}_1)}\otimes m)\big)=
\mu_{state}^T\big(p_1(g)\otimes m + 
1_{T(\f{g}_1)}\otimes \sigma_2 p_2(g).m\big)=\cr
&& \mu_{state}^T(p_1(g)\otimes m) + 
\mu_{state}^T\big(1_{T(\f{g}_1)}\otimes \sigma_2 p_2(g).m\big)=
\sigma(j_1p_1(g)+j_2p_2(g)).\calU(j_2)(m)=\cr
&&\sigma(g).m=
\sigma(g).\mu_{state}^T(1_{T(\f{g}_1)}\otimes 1_{\calU(\f{g}_2)}).\calU(j_2)(m)
\end{eqnarray}
For $n\geq 1$ we prove \eqref{mu-compat} by induction using linear generators of $T_n(\f{g}_1)$ i.e. the family $(x\otimes t)_{x\in \f{g}_1\times T_{n-1}(\f{g}_1)}$ then 
\small{
\begin{eqnarray}
&&\mu_{state}^T\big(g\ast_T((x\otimes t)\otimes m)\big)=
\mu_{state}^T\big([g,j_1(x)]\ast_T(t\otimes m) \big)+ 
\mu_{state}^T\big(x\otimes (g\ast_T(t\otimes m))\big)=\\[2mm]
&& \sigma([g,j_1(x)]).\calU(j_1)\psi_1(t).\calU(j_2)(m) + \calU(j_1)\sigma_1(x).\sigma(g).\calU(j_1)\psi_1(t).\calU(j_2)(m)=\\[2mm]
&&\sigma([g,j_1(x)]).\calU(j_1)\psi_1(t).\calU(j_2)(m) + \sigma j_1(x).\sigma(g).\calU(j_1)\psi_1(t).\calU(j_2)(m)=\\[2mm]
&&\big(\sigma(g).\sigma j_1(x)-\sigma j_1(x).\sigma(g)\big).\calU(j_1)\psi_1(t).\calU(j_2)(m) +\\[2mm] 
&&\sigma j_1(x).\sigma(g).\calU(j_1)\psi_1(t).\calU(j_2)(m)=\\[2mm]
&&\sigma(g).\sigma j_1(x).\calU(j_1)\psi_1(t).\calU(j_2)(m)=
\sigma(g).\calU(j_1)\psi_1(x).\calU(j_1)\psi_1(t).\calU(j_2)(m)=\\[2mm]
&&\sigma(g).\calU(j_1)\big(\psi_1(x).\psi_1(t)\big).\calU(j_2)(m)=
\sigma(g).\calU(j_1)\psi_1(x\otimes t).\calU(j_2)(m)=\\[2mm]
&&\sigma(g).\mu_{state}^T((x\otimes t)\otimes 1_{\calU(\f{g})}).\calU(j_2)(m)
\end{eqnarray}
}
\item[iv)] We first construct $g\ast_U(t\otimes m)$ for tensors of the type 
$t\otimes 1_{\calU(\f{g}_2)}$ i.e. we construct the restriction of $\Phi_U$ on 
$\f{g}\otimes \calU(\f{g}_1)\otimes 1_{\calU(\f{g}_2)}$ and prove that the following diagram commutes 
\begin{equation}\label{Cell1Reduced}
\begin{tikzcd}[column sep=2cm, row sep=0.8cm]
T(\f{g}_1)\otimes 1_{\calU(\f{g}_2)}\arrow[d,"g\ast_T",swap]
\arrow[r,"\psi_1\otimes I_{U_2} "] 
& \calU(\f{g}_1)\otimes 1_{\calU(\f{g}_2)}\arrow[d,"g\ast_U",swap]\\  
T(\f{g}_1)\otimes \calU(\f{g}_2) \arrow[r,"\psi_1\otimes I_{U_2}"] 
& \calU(\f{g}_1)\otimes \calU(\f{g}_2)
\end{tikzcd}
\end{equation} 
and use the following lemma
\begin{lemma}\label{lemma1}
Let $\calA_i,\ \calB,\ i=1..2$ be $\kk$-AAU and 
$s:\ \calA_1\to \calA_2,\ \epsilon:\calB\to \kk$ be morphisms 
(of $\kk$-AAU). Then\\
i) $\calA_i\to \calA_i\otimes \calB$ defined by $x\mapsto x\otimes 1_{\calB}$ is injective (the image of it will be noted $\calA_i\otimes 1_{\calB}$).\\
ii) The kernel of $s\otimes Id_{\calB}$ is $\ker(s)\otimes 1_{\calB}$.  
\end{lemma} 
\begin{proof}
Left to the reader.
\end{proof}
\textbf{End of the proof of (iv)}\pointir\\
We complete the proof of diagram \eqref{Cell1Reduced}. As $\psi_1$ is 
surjective, so is $\psi_1\otimes I_{U_2}$ (even its retriction i.e. from 
$T(\f{g}_1)\otimes 1_{\calU(\f{g}_2)}$ to 
$\calU(\f{g}_1)\otimes 1_{\calU(\f{g}_2)}$, let us call $\xi_1$ this restriction)
so that the diagram \eqref{Cell1Reduced}, in fact, becomes 
\begin{equation}
\begin{tikzcd}[column sep=2cm, row sep=0.8cm]
T(\f{g}_1)\otimes 1_{\calU(\f{g}_2)}\arrow[d,"g\ast_T",swap]
\arrow[r,"\xi_1"] 
& \calU(\f{g}_1)\otimes 1_{\calU(\f{g}_2)}\arrow[d,dotted,"g\ast_U",swap]\\  
T(\f{g}_1)\otimes \calU(\f{g}_2) \arrow[r,"\psi_1\otimes I_{U_2}"] 
& \calU(\f{g}_1)\otimes \calU(\f{g}_2)
\end{tikzcd}
\end{equation} 
 From Lemma \ref{lemma1}, the kernel of $\sigma_1$ is the module generated, for $(p,x,y)\in T(\f{g}_1)\otimes \f{g}_1\otimes \f{g}_1$ by the family of tensors ($s$ is omitted in the indexation because it will not vary throughout the proof  
$$
E(p,x,y):=(p\otimes x\otimes y\otimes s\otimes 1_{\calU(\f{g}_2)})-
(p\otimes y\otimes x\otimes s\otimes 1_{\calU(\f{g}_2)}) - 
(p\otimes [x,y]\otimes s\otimes 1_{\calU(\f{g}_2)})
$$
then, the existence (and unicity) of $g\ast_U-$ amounts to prove that, for\\ 
$(p,x,y)\in T(\f{g}_1)\otimes \f{g}_1\otimes \f{g}_1$, we have  $g\ast_T(E(p,x,y))=0$. Let us set 
$T(p,x,y):=g\ast_T((p\otimes x\otimes y\otimes s\otimes 1_{\calU(\f{g}_2)})$. We proceed by cases.\\
First case : $p\in T_{n}(\f{g}_1)$ for $n\ge1$.\\
We check the fact for the tensors $p=a\otimes p'$ (sufficient because these tensors generate $T_+(\f{g}_1)=\oplus_{n\ge1} T_{n}(\f{g}_1)$. Let us set 
$T(p,u):=p\otimes u\otimes s\otimes 1_{\calU(\f{g}_2)}$, we have to prove that 
$g\ast_T(T(p,x\otimes y))-g\ast_T(T(p,y\otimes x))\equiv g\ast_T(T(p,[x,y])$\\ 
where $X\equiv Y$ stands for $X-Y\in \ker(\psi_1\otimes I_{U_2})$.\\ 
By direct computation we get 
\begin{eqnarray}
&&g\ast_T(a\otimes p'\otimes x\otimes y\otimes s\otimes 1_{\calU(\f{g}_2)})=\cr
&& [g,a]\ast_T(p'\otimes x\otimes y\otimes s\otimes 1_{\calU(\f{g}_2)}) +
a\otimes g\ast_T(p'\otimes x\otimes y\otimes s\otimes 1_{\calU(\f{g}_2)})
\end{eqnarray} 
from this, we see, by induction, that all amounts to prove the fact for $n=0$. Then,\\   
Second case : $p\in T_{n}(\f{g}_1)$ for $n=0$.\\
By homogeneity, we can suppose $p=1_{T_{n}(\f{g}_1)}$. 
Let us compute 
\begin{eqnarray}
&&g\ast_T(x\otimes y\otimes s\otimes 1_{\calU(\f{g}_2)})=\cr
&& \underbrace{[g,x]\ast_T(y\otimes s\otimes 1_{\calU(\f{g}_2)})}_{T_1(x,y)} 
+\underbrace{x\otimes g\ast_T(y\otimes s\otimes 1_{\calU(\f{g}_2)})}_{T_2(x,y)}=\cr 
&& \underbrace{[[g,x],y]\ast_T(s\otimes 1_{\calU(\f{g}_2)})+
y\otimes \big([g,x]\ast_T(s\otimes 1_{\calU(\f{g}_2)})\big)
}_{T_1(x,y)=T_{11}(x,y)+T_{12}(x,y)}\cr
&&+\underbrace{
x\otimes \big([g,y]\ast_T(s\otimes 1_{\calU(\f{g}_2)})\big)+
x\otimes y\otimes \big(g\ast_T(s\otimes 1_{\calU(\f{g}_2)})\big)
}_{T_2(x,y)=T_{21}(x,y)+T_{22}(x,y)}
\end{eqnarray} 
Then 
\begin{eqnarray}
&& T_{11}(x,y) - T_{11}(y,x)=[g,[x,y]]\ast_T(s\otimes 1_{\calU(\f{g}_2)})\cr
&& T_{12}(x,y) - T_{21}(y,x)=y\otimes \big([g,x]\ast_T(s\otimes 1_{\calU(\f{g}_2))}\big) - 
y\otimes \big([g,x]\ast_T(s\otimes 1_{\calU(\f{g}_2)})\big)=0\cr
&& T_{21}(x,y) - T_{12}(y,x)=x\otimes \big([g,y]\ast_T(s\otimes 1_{\calU(\f{g}_2)})\big) - x\otimes \big([g,y]\ast_T(s\otimes 1_{\calU(\f{g}_2)})\big)=0\cr
&& T_{22}(x,y) - T_{22}(y,x)=
x\otimes y\otimes \big(g\ast_T(s\otimes 1_{\calU(\f{g}_2)})\big)
-y\otimes x\otimes \big(g\ast_T(s\otimes 1_{\calU(\f{g}_2)})\big)\equiv \cr
&& [x,y]\otimes \big(g\ast_T(s\otimes 1_{\calU(\f{g}_2)})\big)
\end{eqnarray}
(we recall that $X\equiv Y$ stands for $X-Y\in \ker(\psi_1\otimes I_{U_2})$). 
Then
\begin{eqnarray}
&&g\ast_T(x\otimes y\otimes s\otimes 1_{\calU(\f{g}_2)})-
g\ast_T(y\otimes x\otimes s\otimes 1_{\calU(\f{g}_2)})\equiv \cr
&& [g,[x,y]]\ast_T(s\otimes 1_{\calU(\f{g}_2)})+ 
[x,y]\otimes \big(g\ast_T(s\otimes 1_{\calU(\f{g}_2)})\big)=\cr
&& g\ast_T([x,y]\otimes s\otimes 1_{\calU(\f{g}_2)})
\end{eqnarray}
then, there exists $g\ast_U$ such that \eqref{Cell1Reduced} commutes.\\
\textbf{End of the proof of \eqref{Cell1}}\pointir\\
We set, for $(g,t,m)\in \f{g}\times \calU(\f{g}_1)\times \calU(\f{g}_2)$,  
\begin{equation}
g\ast_U(t\otimes m):=g\ast_U(t\otimes 1_{\calU(\f{g}_2)}).m
\end{equation}
Now, we remark that $(g,t,m)\mapsto g\ast_U(t\otimes m)$ is trilinear and this completes the proof.
\end{enumerate}
\end{proof}
\begin{corollary}\label{mu_U-state-compat} 
For all $(g,m_1,m_2)\in \f{g}\times \calU(\f{g}_1)\times \calU(\f{g}_2)$, one has 
\begin{equation}\label{mu_U-state-compat-eq} 
\mu_{state}^U(g\ast_U(m_1\otimes m_2))=\sigma(g).\,\calU(j_1)(m_1).\,\calU(j_2)(m_2)
\end{equation}
\end{corollary}
\begin{proof}
From theorem \eqref{Th-A} (iii.b, in particular \eqref{mu-compat}) and diagram \eqref{Cell1}. 
\end{proof}

\subsection{$g\ast_U$ is a $\f{g}$-action on 
$\calU(\f{g}_1)\otimes \calU(\f{g}_2)$.}\label{Lie action}
We here prove that $g\ast_U$ defines a Lie $\f{g}$-action on $\calU(\f{g}_1)\otimes \calU(\f{g}_2)$ i.e. for all $(g,h,m_1,m_2)\in \f{g}_1\times \f{g}_1\times \calU(\f{g}_1)\times \calU(\f{g}_2)$ we have (below $\ast$ will stand 
for $\ast_U$)
\begin{equation}\label{eq1}
g\ast\big(h\ast (m_1\otimes m_2)\big)-
h\ast\big(g\ast (m_1\otimes m_2)\big)=[g,h]\ast (m_1\otimes m_2)
\end{equation}
Let us then set $T(g,h)=g\ast\big(h\ast (m_1\otimes m_2)\big)$. \\
We have 4 cases (which can be reduced to 3 by antisymmetry)\\
a) $(g,h)\in \f{g}_1\times \f{g}_1$\\
\begin{eqnarray}
&& T(g,h)-T(h,g)=g\ast(h\ast (m_1\otimes m_2))-
h\ast(g\ast (m_1\otimes m_2))=\cr
&& \sigma_1(g).\sigma_1(h).m_1\otimes m_2 - \sigma_1(h).\sigma_1(g).m_1\otimes m_2\equiv\cr
&&\sigma_1([g,h]).m_1\otimes m_2=
[g,h]\ast (m_1\otimes m_2)
\end{eqnarray}
b) $(g,h)\in \f{g}_2\times \f{g}_1$\\
\begin{eqnarray}
&& T(g,h)-T(h,g)=g\ast(h\ast (m_1\otimes m_2))-h\ast(g\ast (m_1\otimes m_2))=\cr
&& g\ast(\sigma_1(h).m_1\otimes m_2)-\sigma_1(h).(g\ast (m_1\otimes m_2))=\cr
&& [g,j_1(h)]\ast (m_1\otimes m_2)+\sigma_1(h).(g\ast (m_1\otimes m_2))-
\sigma_1(h).(g\ast (m_1\otimes m_2))=\cr
&& [g,j_1(h)]\ast (m_1\otimes m_2)
\end{eqnarray}
c) $(g,h)\in \f{g}_1\times \f{g}_2$\\
Is true by antisymmetry.\\
d) $(g,h)\in \f{g}_2\times \f{g}_2$\\
For the computation of $T(g,h)-T(h,g)$, we have two cases.\\
d1) $m_1=1_{\calU(\f{g}_1)}$\\
\begin{eqnarray}
&& T(g,h)-T(h,g)=g\ast(h\ast (1_{\calU(\f{g}_1)}\otimes m_2))-h\ast(g\ast (1_{\calU(\f{g}_1)}\otimes m_2))=\cr
&& 1_{\calU(\f{g}_1)}\otimes \sigma_2(g).\sigma_2(h).m_2-1_{\calU(\f{g}_1)}\otimes \sigma_2(h).\sigma_2(g).m_2=
1_{\calU(\f{g}_1)}\otimes \sigma_2([g,h]).m_2=
[g,h]\ast (1_{\calU(\f{g}_1)}\otimes m_2)
\end{eqnarray} 
d2) $m_1\in \calU_+(\f{g}_1)$\\
We prove \eqref{eq1} by induction. Let $m_1\in \calU_n(\f{g}_1)$.\\
We have $n\ge1$ and $\calU_n(\f{g}_1)$ is generated by the products 
$x.m$ with $x\in \calU(\f{g}_1)$ and $m\in \calU_{n-1}(\f{g}_1)$ 
\begin{eqnarray}%
&& T(g,h)-T(h,g)=g\ast(h\ast (x.m\otimes m_2))-h\ast(g\ast (x.m\otimes m_2))=\cr
&& g\ast([h,x]\ast (m\otimes m_2))+g\ast(\sigma_1(x).(h\ast (m\otimes m_2)))\cr
&& - h\ast([g,x]\ast (m\otimes m_2))-h\ast(\sigma_1(x).(g\ast (m\otimes m_2)))\cr
&& = \underbrace{g\ast([h,x]\ast (m\otimes m_2))}_{T_1(g,h)} + 
\underbrace{[g,x]\ast(h\ast (m\otimes m_2)))}_{T_2(g,h)} +
\underbrace{\sigma_1(x).(g\ast (h\ast (m\otimes m_2)))}_{T_3(g,h)}\cr
&& - \underbrace{h\ast([g,x]\ast (m\otimes m_2))}_{T_1(h,g)}
 - \underbrace{[h,x]\ast(g\ast (m\otimes m_2)))}_{T_2(h,g)} 
 - \underbrace{\sigma_1(x).(h\ast (g\ast (m\otimes m_2)))}_{T_3(h,g)}
\end{eqnarray}
Then
\begin{eqnarray}
&& T_1(g,h)-T_2(h,g)=[g,[h,x]]\ast (m\otimes m_2)\mbox{ by induction}\cr
&& T_2(g,h)-T_1(h,g)=[[g,x],h]\ast (m\otimes m_2)\mbox{ by induction}\cr
&& T_3(g,h)-T_3(h,g)=\sigma_1(x).([g,h]\ast (m\otimes m_2))\mbox{ by induction}\cr
&& \mbox{ Hence }T(g,h)-T(h,g)=([g,[h,x]]+[[g,x],h])\ast (m\otimes m_2) + 
x.([g,h]\ast (m\otimes m_2))\cr
&& = [[g,h],x]\ast (m\otimes m_2) + \sigma_1(x).([g,h]\ast (m\otimes m_2)) \cr
&& = [g,h]\ast (x.m\otimes m_2)
\end{eqnarray}
We now come to the proof that $\mu_{state}^U$ is one-to-one. 
\section{The linear map $\mu_{state}^U$ is bijective.}
\begin{theorem}
i) From the (Lie) action $\ast_U$, one deduces a unique $\calU(\f{g})-$ module structure on $\calU(\f{g}_1)\otimes \calU(\f{g}_2)$ (noted $\ast_{mod}$) such that 
$\sigma(g)\ast_{mod}(m_1\otimes m_2)=g\ast_{\calU}(m_1\otimes m_2)$.\\
ii) The map $s:\ m\mapsto m\ast_{mod}(1_{\calU(\f{g}_1)}\otimes 1_{\calU(\f{g}_2)})$ and $\mu_{state}^U$ are mutually inverse.
\end{theorem}
\begin{proof}
i) From Theorem \eqref{Th-A} (iv), let us note (as above) $g\ast_{\calU}-$, the map 
$m_1\otimes m_2\mapsto g\ast_{\calU}(m_1\otimes m_2)$, we then get a linear map 
$\varphi:\ \f{g}\to End\left(\calU(\f{g}_1)\otimes \calU(\f{g}_2)\right)$ and, by 
\eqref{eq1}, we learn that $\varphi$ is a morphism of $\kk$-Lie algebras. 
By universal property of $\calU(\f{g})$, we get 
\begin{equation}\label{UnivDiag}
\begin{tikzcd}[column sep=1.6cm, row sep=0.8cm]
\Lie{\kk}\arrow[rr,leftarrow,"F",pos=0.4,densely dotted]
    \arrow[dd,dash,
    start anchor={[xshift=17ex, yshift=4ex]},
    end anchor={[xshift=17ex]}
    ]&
& \AAU{\kk}\\[-20pt] 
\f{g} \arrow[rr,"\varphi",pos=0.3,dashed]
\arrow[rrd,dashed, "\sigma",pos=0.41,swap]&      &  
End\left(\calU(\f{g}_1)\otimes \calU(\f{g}_2)\right)\\
\phantom{x}&& \calU(\mathfrak{g}). \arrow[u, "\widehat{\varphi}"']
\end{tikzcd}
\end{equation} 
which means that, for all $(g,m_1,m_2)\in \f{g}\times\calU(\f{g}_1)\times \calU(\f{g}_2)$, 
\begin{equation}
\varphi(g)[m_1\otimes m_2]=\widehat{\varphi}(\sigma(g))[m_1\otimes m_2]
\end{equation}
Of course, such a morphism as $\widehat{\varphi}$ defines at once a structure of left $\calU(\mathfrak{g})$-module on $\calU(\f{g}_1)\otimes \calU(\f{g}_2)$. Its action will be noted $\ast_{mod}$ such that
$$
\sigma(g)\ast_{mod}(m_1\otimes m_2):=\widehat{\varphi}(\sigma(g))[m_1\otimes m_2]
$$
which completes the first point.\\
ii) Knowing already that $\mu_{state}^U$ is surjective, it will be sufficient to establish that 
$$
s\circ \mu_{state}^U=Id_{\calU(\f{g}_1)\otimes\, \calU(\f{g}_2)}
$$    
which amounts to show that for $(g_i)_{1\leq i\leq p}$ in $\f{g}_1$ (resp.
$(h_i)_{1\leq i\leq q}$ in $\f{g}_2$) 
\begin{equation}\label{s-mu}
s\circ \mu_{state}^U\big(
\sigma_1(g_1)\cdots \sigma_1(g_p) \otimes \sigma_2(h_1)\cdots \sigma_2(h_q) \big)=
\sigma_1(g_1)\cdots \sigma_1(g_p) \otimes \sigma_2(h_1)\cdots \sigma_2(h_q) 
\end{equation}
By linearity, this will prove that 
$s\circ \mu_{state}^U=Id_{\calU(\f{g}_1)\otimes\, \calU(\f{g}_2)}$.\\ 
From \eqref{mu_U-state-compat-eq}, for $p>0$, we get, 
\begin{eqnarray}
&&\mu_{state}^U\big(
\sigma_1(g_1)\cdots \sigma_1(g_p) \otimes \sigma_2(h_1)\cdots \sigma_2(h_q) \big)=\cr
&&\calU(j_1)\sigma_1(g_1)\mu_{state}^U\big(
\sigma_2(g_2)\cdots \sigma_1(g_p) \otimes \sigma_2(h_1)\cdots \sigma_2(h_q) \big)
\end{eqnarray}
and, remarking that $s$ is $\calU(\f{g})-$ linear we have 
\begin{eqnarray}
&& s\circ \mu_{state}^U\big(
\sigma_1(g_1)\cdots \sigma_1(g_p) \otimes \sigma_2(h_1)\cdots \sigma_2(h_q) \big)=\cr
&& s\Big(
\calU(j_1)\sigma_1(g_1)\mu_{state}^U\big(
\sigma_2(g_2)\cdots \sigma_1(g_p) \otimes \sigma_2(h_1)\cdots \sigma_2(h_q) \big)\Big)=\cr
&& s\Big(
\sigma j_1(g_1)\mu_{state}^U\big(
\sigma_2(g_2)\cdots \sigma_1(g_p) \otimes \sigma_2(h_1)\cdots \sigma_2(h_q) \big)\Big)=\cr
&& j_1(g_1).s\mu_{state}^U\big(
\sigma_2(g_2)\cdots \sigma_1(g_p) \otimes \sigma_2(h_1)\cdots \sigma_2(h_q) \big)=\cr
&& \sigma_1(g_1).\cdots .\sigma_1(g_p)\otimes \sigma_2(h_1).\cdots .\sigma_2(h_q) 
\end{eqnarray}
the other case ($p=0$) is straightforward. Then, by induction on $p$, one has 
$s\circ \mu_{state}^U(m_1\otimes m_2)=m_1\otimes m_2$ which proves the claim.

\smallskip
QED
\end{proof}

%
\section{Conclusion and future}
\begin{enumerate}
\item Quantized enveloping algebras
\item Lie superalgebras
\end{enumerate}
\newpage

\begin{thebibliography}{9999999999999}                                                                                   
%
\bibitem{B-ToS}{N. Bourbaki}, \textit{Theory of Sets}, Springer-Verlag Berlin Heidelberg New York; (2nd printing 1989).
%
\bibitem{B-Alg13}\textsc{N. Bourbaki}, \textit{Algebra I (Chapters 1-3)}, Springer 1989.
%
\bibitem{B-Lie13} N. Bourbaki, \textit{Lie groups and Lie algebras, Chapters 1-3},  Springer-Verlag; (1989).
%
\bibitem{PC1} Cartier P.,\textit{Vinberg algebras, Lie groups and combinatorics}, Quanta of Maths, Clay mathematics proceedings; vol. 11 (January 2011)
%
\bibitem{NDVU-PhD} Dinh (Vu Nguyen), \textit{Combinatorics of Lazard Elimination and Interactions},\\
\url{https://hal.science/tel-04367964} and (more updated)\\
\url{https://www-lipn.univ-paris13.fr/~duchamp/Theses/Vu/}
%
\bibitem{GriRei} D. Grinberg and V. Reiner, \textit{Hopf algebras
in Combinatorics}, ArXiv version,\\
\url{http://www.arxiv.org/abs/1409.8356v7}{\texttt{arXiv:1409.8356v7}}.
\newline See also, for the free algebra of NonCommutative polynomials, Exercise 1.6.8 p31, there\\ 
\url{http://www.cip.ifi.lmu.de/~grinberg/algebra/HopfComb-sols.pdf}\\ 
(More frequently updated version).
%
\bibitem{VK} Kac V., \textit{Infinite dimensional Lie algebras}
%
\bibitem{LothCOW} M. Lothaire, \textit{Combinatorics on words}, Cambridge University Press (2003).
%
\bibitem{McLane} S. Mac Lane, \textit{Categories for the Working Mathematician}, second edition, Graduate Texts in Mathematics, vol. 5, Springer-Verlag, New York; (1998).
%
\bibitem{MP} Robert V. Moody, Arturo Pianzola, \textit{Lie Algebras with Triangular Decompositions}, ISBN: 978-0-471-63304-4 April 1995
%
\bibitem{FLA} C. Reutenauer, \textit{Free Lie Algebras}, Universit\'e du Qu\'ebec à Montr\'eal, Clarendon Press, Oxford; (1993).
%
\bibitem{Hets} Heteromorphism,\\
\url{https://ncatlab.org/nlab/show/heteromorphism} 
%
\bibitem{MO-Orig} About Enveloping Algebras of direct sums\\
\url{https://mathoverflow.net/questions/300851}
%
\bibitem{MO-Smash} Could we define the semi direct product of two universal enveloping algebras ?\\
\url{https://mathoverflow.net/questions/142623}
%
\bibitem{MO96078} Are semi-direct products categorical colimits,\\
\url{https://mathoverflow.net/questions/96078}
%
\bibitem{MaxPlus} Tropical semiring,\\
\url{https://en.wikipedia.org/wiki/Tropical_semiring}
%
\bibitem{TransStruct} Transport of structure,\\
\url{https://en.wikipedia.org/wiki/Transport_of_structure}
\end{thebibliography}
\end{document}